\nonstopmode

\documentclass[letter,11pt,en]{quick-document}

\usepackage{authblk}
\usepackage{fullpage}
\usepackage{tikz, pgf}
\usepackage{tikz-qtree}
\usepackage{enumitem}
\usepackage[disable]{todonotes}

\usepackage{listings}

\lstdefinestyle{numbered}{
  frame=L,
  xleftmargin=\parindent,
  numbers = left
language=C,
showstringspaces=false,
basicstyle=\footnotesize\ttfamily,
keywordstyle=\bfseries\color{green!40!black},
commentstyle=\itshape\color{purple!40!black},
identifierstyle=\color{blue},
stringstyle=\color{orange},
}

\ifcsname algorithm\endcsname\relax\else%
\RequirePackage{algpseudocode, algorithm}
\def\OneLineIf#1#2{\State\algorithmicif\ #1\ \algorithmicthen\ #2}

\algnewcommand\algorithmicbreak{\textbf{break}}
\algnewcommand\algorithmicto{\textbf{to}}
\algrenewtext{For}[3]%
{\algorithmicfor\ #1 = #2 \algorithmicto\ #3 \algorithmicdo}
\algnewcommand\Break{\algorithmicbreak{} }
\fi%

\usepackage{hyperref}

  \title{\textsc{MergeShuffle}:\\A Very Fast, Parallel Random Permutation
    Algorithm}

\author{Axel Bacher%
  \thanks{Electronic address: \texttt{axel.bacher@lipn.univ-paris13.fr}}~}
\author{Olivier Bodini%
  \thanks{Electronic address: \texttt{olivier.bodini@lipn.univ-paris13.fr}}~}
\author{Alexandros Hollender%
  \thanks{Electronic address: \texttt{halexandros@web.de}}~}
\author{J\'er\'emie Lumbroso%
  \thanks{Electronic address: \texttt{lumbroso@cs.princeton.edu};
    corresponding author}}
\affil{}

\date{\today}

\begin{document}

\maketitle

\begin{abstract}
  This article introduces an algorithm, \textsc{MergeShuffle}, which is an
  extremely efficient algorithm to generate random permutations (or to
  randomly permute an existing array). It is easy to implement, runs in
  $n\cramped{\log_2} n + O(1)$ time, is in-place, uses
  $n\cramped{\log_2} n + \Theta(n)$ random bits, and can be parallelized
  accross any number of processes, in a shared-memory PRAM model. Finally,
  our preliminary simulations using OpenMP\footnote{Full code available
    at: \url{https://github.com/axel-bacher/mergeshuffle}}
  suggest it is more efficient than the Rao-Sandelius algorithm, one of
  the fastest existing random permutation algorithms.

  We also show how it is possible to further reduce the number of random
  bits consumed, by introducing a second algorithm
  \textsc{BalancedShuffle}, a variant of the Rao-Sandelius algorithm which
  is more conservative in the way it recursively partitions arrays to be
  shuffled. While this algorithm is of lesser practical interest, we
  believe it may be of theoretical value.
\end{abstract}

Random permutations are a basic combinatorial object, which are useful in
their own right for a lot of applications, but also are usually the
starting point in the generation of other combinatorial objects, notably
through bijections.

The well-known Fisher-Yates shuffle~\cite{FiYa48,
  Durstenfeld64} iterates through a sequence from the end to the beginning
(or the other way) and for each location $i$, it swaps the value at $i$
with the value at a random target location $j$ at or before $i$. This
algorithm requires very few steps---indeed a random integer and a swap at
each iteration---and so its efficiency and simplicity have until now stood
the test of time.

\begin{algorithm}
  \caption{The classical Fisher-Yates shuffle~\cite{FiYa48} to generate
    random permutations, as per
    Durstenfeld~\cite{Durstenfeld64}.\label{alg-fy-shuffle}}
\begin{algorithmic}[1]
\Procedure{FisherYatesShuffle}{$T$}
\For{$i$}{$n-1$}{$0$}
  \State $j \gets$ random integer from $\{0, \ldots, i\}$
  \State \Call{Swap}{$T$, $i$, $j$}
\EndFor
\EndProcedure
\end{algorithmic}
\end{algorithm}

But there have been two trends in trying to improve this algorithm: first,
initially the algorithm assumes some source of randomness that allows for
discrete uniform variables, but this there has been a shift towards
measuring randomness better with the random bit model; second, with the
avent of large core clusters and GPUs, there is an interest in making
parallel versions of this algorithm.

\paragraph{The random-bit model.}

Much research has gone into simulating probability distributions, with
most algorithms designed using infinitely precise \emph{continuous uniform
  random} variables (see \cite[II.3.7]{Devroye86}). But because
(pseudo-)randomness on computers is typically provided as 32-bit
integers---and even bypassing issues of true randomness and bias---this
model is questionable. Indeed as these integers have a fixed precision,
two questions arise: when are they not precise enough? when are they too
precise? These are questions which are usually ignored in typical
fixed-precision implementations of the aforementioned algorithms. And it
suggests the usefulness of a model where the unit of randomness is not the
uniform random variable, but the \emph{random bit}.

This random bit model was first suggested by Von Neumann~\cite{Neumann51},
who humorously objected to the use of fixed-precision pseudo-random
uniform variates in conjunction with transcendant functions approximated
by truncated series. His remarks and algorithms spurred a fruitful line of
theoretical research seeking to determine \emph{which} probabilities can
be simulated using only random bits (unbiased or biased? with known or
unknown bias?), with which complexity (expected number of bits used?), and
which guarantees (finite or infinite algorithms? exponential or
heavy-tailed time distribution?). Within the context of this article, we
will focus on designing practical algorithms using unbiased random bits.

In 1976, Knuth and Yao~\cite{KnYa76} provided a rigorous theoretical
framework, which described generic optimal algorithms able to simulate any
distribution. These algorithms were generally not practically usable:
their description was made as an infinite tree---infinite not only in the
sense that the algorithm terminates with probability $1$ (an unavoidable
fact for any probability that does not have a finite binary expansion),
but also in the sense that the description of the tree is infinite and
requires an infinite precision arithmetic to calculate the binary
expansion of the probabilities.

In 1997, Han and Hoshi~\cite{HaHo97} provided the \emph{interval
  algorithm}, which can be seen as both a generalization and
implementation of Knuth and Yao's model. Using a random bit stream, this
algorithm amounts to simulating a probability $p$ by doing a binary search
in the unit interval: splitting the main interval into two equal
subintervals and recurse into the subinterval which contains $p$. This
approach naturally extends to splitting the interval in more than two
subintervals, not necessarily equal. Unlike Knuth and Yao's model, the
interval algorithm is a concrete algorithm which can be readily
programmed... as long as you have access to arbitrary precision arithmetic
(since the interval can be split to arbitrarily small sizes). This work
has recently been extended and generalized by Devroye and
Gravel~\cite{DeGr15}.

We were introduced to this problematic through the work of Flajolet,
Pelletier and Soria~\cite{FlPeSo11} on \emph{Buffon machines}, which are a
framework of probabilistic algorithms allowing to simulate a wide range of
probabilities using only a source of random bits.

One easy optimization of the Fisher-Yates algorithm (which we use in our
simulations) is to use an recently discovered optimal way of drawing
discrete uniform variables~\cite{Lumbroso13}.

\paragraph{Prior Work in Parallelization.}

There has been in particular a great deal of interest in finding efficient
parallel algorithms to randomly generate permutations, in various many
contexts of parallelization, some theoretical and some
practical~\cite{Gustedt03, Gustedt08, Sanders98, Hagerup91, Alonso96,
  CoBa06, CzKaKuLo98, Anderson90}.

Most recently, Shun~\etal~\cite{ShGuBlFiGi15} wrote an enlightening
article, in which they looked at the intrinsic parallelism inherent in
classical sequential algorithms, and these can be broken down into
independent parts which may be executed separately. One of the algorithms
they studied is the Fisher-Yates shuffle. They considered the insertion of
each element of the algorithm as a separate part, and showed that the
dependency graph, which provides the order in which the parts must be
executed, is a random binary search tree, and as such, is well known to
have on average a logarithmic height~\cite{Devroye86}. This allowed them
to show that the algorithm could be distributed on $n/\log n$ processors.

Because they aimed for generality (and designed a framework to adapt other
similar sequential algorithms), their resulting algorithm is not as
optimized as can be.

We believe our contribution improves on this work by providing a parallel
algorithm with similar guarantees, and which runs, in practice, extremely
fast.

\begin{algorithm}
\caption{The \textsc{MergeShuffle}{} algorithm.\label{alg-mergeshuffle}}
\begin{algorithmic}[1]
\Procedure{MergeShuffle}{$T$, $k$}%
\Comment{$k$ is the cut-off threshold at which to shuffle with Fisher-Yates.}
\State Divide $T$ into $\cramped{2^k}$ blocks of roughly the same size
\State Shuffle each block independently using the Fisher-Yates shuffle
\State $p \gets k$
\Repeat
  \State Use the \Call{Merge}{} procedure to merge adjacent blocks of size
  $\cramped{2^p}$ into new blocks of size $\cramped{2^{p+1}}$
  \State $p \gets p + 1$
\Until{$T$ consists of a single block}
\EndProcedure
\end{algorithmic}
\end{algorithm}

\paragraph{Splitting Processes.} Relatively recently,
Flajolet~\etal~\cite{FlPeSo11} formulated an elegant random permutation
algorithm which uses only random bits, using the \emph{trie} data
structure, which models a splitting process: associate to each element of
a set $x\in S$ an \emph{infinite} random binary word $\cramped{w_x}$, and
then insert the key-value pairs $(\cramped{w_x}, x)$ into the trie; the
ordering provided by the leaves is then a random permutation.

This general concept is elegant, and it is optimized in two ways:
\begin{itemize}[noitemsep]
\item the binary words thus do not need to be infinite, but only long
  enough to completely distinguish the elements;
\item the binary words do not need to be drawn \emph{a~priori}, but may be
  drawn one bit (at each level of the trie) at a time, until each element
  is in a leaf of its own.
\end{itemize}
This algorithm turns out to have been already exposed in some form in the
early 60's, independently by Rao~\cite{Rao61} and by
Sandelius~\cite{Sandelius62}. Their generalization extends to the case
where we split the set into $R$ subsets (and where we would then draw
random integers instead of random bits), but in practice the case $R=2$ is
the most efficient. The interest of this algorithm is that it is, as far
as we know, the first example of a random permutation algorithm which was
written to be parallelized.

\section{The \textsc{MergeShuffle} algorithm}

The new algorithm which is the central focus of this paper was designed by
progressively optimizing a splitting-type idea for generating random
permutation which we discovered in Flajolet~\etal~\cite{FlPeSo11}. The
resulting algorithm closely mimics the structure and behavior of the
beloved \textsc{MergeSort} algorithm. It gets the same guarantees as this
sorting algorithm, in particular with respect to running time and being in-place.

To optimize the execution of this algorithm, we also set a cut-off
threshold, a size below which permutations are shuffled using the
Fisher-Yates shuffle instead of increasingly smaller recursive calls. This
is an optimization similar in spirit to that of \textsc{MergeSort}, in
which an auxiliary sorting algorithm is used on small instances.

\subsection{In-Place Shuffled Merging}

The following algorithm is the linchpin of the MergeShuffle algorithm. It
is a procedure that takes two arrays (or rather, two adjacent ranges of an
array $T$), both of which are assumed to be randomly shuffled, and
produces a shuffled union.

Importantly, this algorithm uses very few bits. Assuming a two equal-sized
sub-arrays of size $k$ each, the algorithm requires
$2k + \Theta(\sqrt{k}\log k)$ random bits, and is extremely efficient in
time because it requires no auxiliary space. (We show an a

\begin{algorithm}
\caption{In-place shuffled merging of two random sub-arrays.}
\label{alg-flip-merge}
\begin{algorithmic}[1]
\Procedure{Merge}{$T$, $s$, $\cramped{n_{1}}$, $\cramped{n_2}$}
\State $i \gets s$
\Comment{$i$, $j$, $n$ are the beginning, middle, and end position
  considered in the array.}
\State $j \gets s+\cramped{n_{1}}$
\State $n \gets s+\cramped{n_{1}}+\cramped{n_{2}}$
\Loop
  \If{\Call{Flip}{{}} = 0}
  \Comment{Flip a coin to determine which sub-array to take an element from.}

    \OneLineIf{$i=j$}{\Break}
  \Else
    \OneLineIf{$j=n$}{\Break}
    \State \Call{Swap}{$T$, $i$, $j$}
    \State $j \gets j + 1$
  \EndIf
  \State $i \gets i + 1$
\EndLoop
\While{$i < n$}
\Comment{One list is depleted; use Fisher-Yates to finish merging.}
  \State Draw a random integer $m \in \{s, \ldots, i\}$
  \State \Call{Swap}{$T$, $i$, $m$}
  \State $i \gets i + 1$
\EndWhile
\EndProcedure
\end{algorithmic}
\end{algorithm}

\begin{lemma}
  Let $A$ and $B$ be two randomly shuffled arrays, respectively of sizes
  $\cramped{n_1}$ and $\cramped{n_2}$. Then the procedure \textsc{Merge}
  produces a randomly shuffled union $C$ of these arrays, of size
  $n=\cramped{n_1}+\cramped{n_2}$.
\end{lemma}

\begin{proof}
  For every integer $k\geqslant 0$, let $\cramped{A_k}$ be the event that,
  after the execution of the first loop (lines~5 to~14) of the procedure
  \Call{Merge}{}, $k$ elements of the list~$A$ remain ($j = n$ and
  $i = n-k$). Similarly, let $\cramped{B_k}$ be the event that $k$
  elements of the list~$B$ remain ($i = j = n-k$). We prove that,
  conditionally to every $\cramped{A_k}$ and $\cramped{B_k}$, the array is
  randomly shuffled after the procedure. We can then conclude from Bayes's
  theorem shows that this is also true unconditionally.

  Let $k \geqslant 0$ and condition by the event $\cramped{A_k}$ (the case
  of $\cramped{B_k}$ is identical). After the execution of the first loop,
  the $n-k$ first elements of the array consist of: $\cramped{n_1} - k$
  elements of $A$; and all $n_2$ elements of~$B$. Let $w$ be the word
  composed of the $n-k+1$ random bits drawn by the first loop. The
  word~$w$ ends with a~$1$ (this bit corresponds to picking an element
  from $B$, which, at that point is depleted, causing the loop to be
  exited). Among the remaining $n-k$ bits, $\cramped{n_1}-k$ are $0$'s and
  $\cramped{n_2}$ are $1$'s, and for $0\leqslant i < n-k$, the element
  $C[i]$ is from~$A$ if $w_i = 0$ and from $B$ otherwise. Since~$A$ and
  $B$ are randomly shuffled and since all words~$w$ are drawn with equal
  probability, this implies that the first~$n-k$ elements are randomly
  shuffled.

  Finally, we use the following loop invariant, which is the same loop
  invariant as in the proof of the Fisher-Yates algorithm: after every
  execution of the second loop (lines~15 to~19), the first $i$ elements of
  the array are randomly shuffled. This shows that the array is randomly
  shuffled after the whole procedure.
\end{proof}

%%% \begin{lemma}
%%%   Let $A$ and $B$ be two randomly shuffled arrays; for simplicity suppose
%%%   they have equal size~$k$. The procedure \Call{Merge}{} produces a
%%%   shuffled array $C$ of size $2k$ using $2k + \Theta(\sqrt{k}\log k)$
%%%   random bits on average.%
%%%   \todo{Par manque de temps, je n'ai pas pu réécrire cet énoncé pour
%%%     $n = \cramped{n_1}+\cramped{n_2}$, et donc j'ai utilisé l'artifice
%%%     lourdingue de dire qu'on suppose $A$ et $B$ de même taille.}
%%% \end{lemma}
%%% 
%%% \begin{proof}
%%%   The number of random bits used depends again, on the size $m$ of the
%%%   word~$w$ drawn during the first loop (lines~5 to~14). Indeed for $m-1$
%%%   of the elements of $C$, we will have shuffled them using only a single
%%%   bit; for the remaining $2k - m +1$ elements, we must insert them in $C$
%%%   by drawing random integer of increasing range $m$, ..., $2k$,
%%%   \begin{align*}
%%%     \binom{2k}{k}\frac{1}{2^{2k}}2k +
%%%     \sum_{m=k}^{2k-1} 2\binom{m}{k}\frac{1}{2^{m+1}}
%%%     \left(m + \sum_{i=m}^{2k} \lceil\log_2 i\rceil\right)\text{.}
%%%   \end{align*}
%%%   We can then bound the last term,
%%%   \begin{align*}
%%%     (2k-m)\log_2 k \leqslant \sum_{i=m+1}^{2k} \log_2 i \leqslant (2k-m)\log_2(2k)
%%%   \end{align*}
%%%   which allows to conclude that the average number of bitsis
%%%   \begin{align*}
%%%     2k + \frac{2}{\sqrt{\pi}}\log_2(k)\sqrt{k} + O(\sqrt{k}).
%%%   \end{align*}
%%% \end{proof}

\begin{lemma}
  The procedure \textsc{Merge} produces a shuffled array $C$ of size $n$
  using $n + \Theta(\sqrt{n}\log n)$ random bits.
\end{lemma}

\begin{proof}
The number of random bits used depends again, on the size $m$ of the
word~$w$ drawn during the first loop (lines~5 to~14). Indeed for $m-1$
of the elements of $C$, we will have shuffled them using only a single
bit; for the remaining $2k - m +1$ elements, we must insert them in $C$
by drawing random integer of increasing range $m$, ..., $2k$,

The number of random bits used only depends on the number of times the
first loop (lines~5 to~14) is executed. Indeed for $m-1$ of the elements
of $C$, we will have shuffled them using only a single bit, and for the
remaining $n - m +1$ elements, we must insert them in $C$ by drawing
random integer of increasing range $m$, ..., $n$. The overall average
number of random bits used is
\begin{align*}
m + \sum_{k=m}^n \lceil\log_2 k\rceil.
\end{align*}
The first $m$ bits are used during the first loop and the rest are used to
draw discrete uniform laws during the second loop.

The first loop stops either because we have drawn $n_1+1$ $0$'s or $n_2+1$
$1$'s (whichever occurs first). In the first case, the average number of
random bits used is thus
\begin{align*}
\sum_{i=0}^{n_2} \binom{n_1+i}{n_1} \frac{1}{2^{n_1+1+i}} \left(n_1+1+i + \sum_{k=n_1+1+i}^n \lceil\log_2 k\rceil\right).
\end{align*}
In this expression $i$ represents the number of $1$'s that were drawn before the $(n_1+1)^{th}$ $0$ was drawn.

Similarly we obtain the following expression for the second case
\begin{align*}
\sum_{i=0}^{n_1} \binom{n_2+i}{n_2} \frac{1}{2^{n_2+1+i}} \left(n_2+1+i + \sum_{k=n_2+1+i}^n \lceil\log_2 k\rceil\right).
\end{align*}

The sum of those two expressions gives the average number of random bits used by the algorithm. By using the following upper and lower bound
\begin{align*}
\log_2(m)(n-m+1) \leq \sum_{k=m}^n \lceil\log_2 k\rceil \leq \log_2(n)(n-m+1)
\end{align*}
we obtain the following asymptotic behaviour for the average number of random bits used
\begin{align*}
n + \Theta(\sqrt{n}\log n).
\end{align*}
\end{proof}

\subsection{Average number of random bits of \textsc{MergeShuffle}}

We now give an estimate of the average number of random bits used by our
algorithm to sample a random permutation of size $n$. Let $cost(k)$ denote
the average number of random bits used by a merge operation with an output
of size $k$. For the sake of simplicity, assume that we sample a random
permutation of size $n = 2^m$. The average number of random bits used is
then
\begin{align*}%\label{bitcomplexity}
\sum_{i=1}^m 2^{m-i} cost(2^i).
\end{align*}

We have seen that $cost(k) = k + \Theta (\sqrt{k} \log k)$. Thus, the
average number of random bits used to sample a random permutation of size
$n = 2^m$ is
\begin{align*}
\sum_{i=1}^m 2^{m-i} \left(2^i + \Theta \left(\sqrt{2^i} \log (2^i)\right)\right) = m 2^m + \Theta \left(2^m \sum_{i=1}^m \frac{i}{2^{i/2}}\right)
\end{align*}
which finally yields
\begin{equation*}
m 2^m + \Theta (2^m) = n \log_2 n + \Theta (n).
\end{equation*}

\section{BalancedShuffle}

For theoretical value, we also present a second algorithm, which
introduces an optimization which be believe has some worth.

\begin{bigcenter}
\includegraphics[scale=0.6]{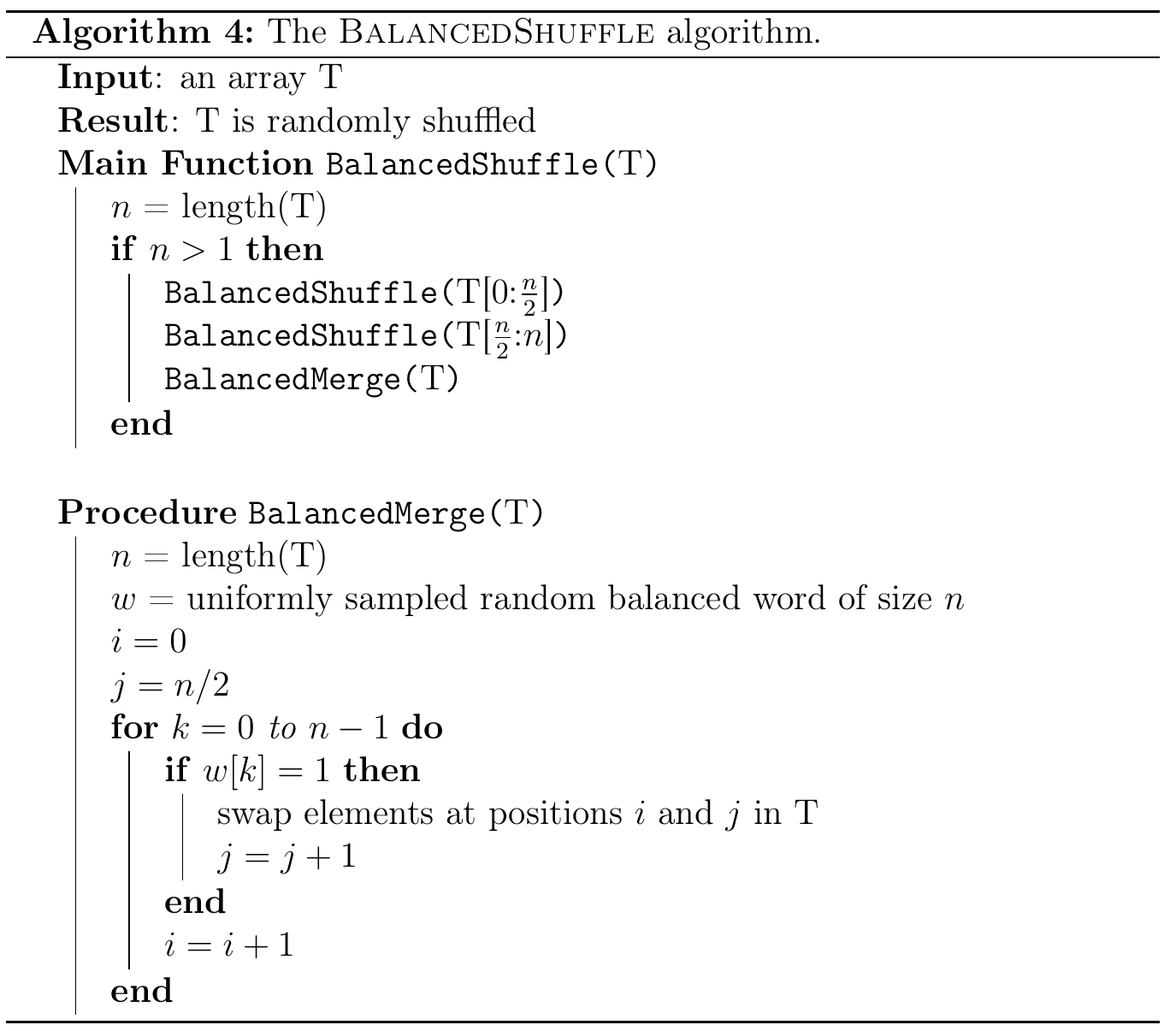}
\end{bigcenter}

\subsection{Balanced Word}

Inspired by Remy~\cite{Remy85}'s now classical and efficient algorithm to
generate random binary of exact size from the repeated drawing of random
integers, Bacher~\etal~\cite{BaBoJa14} produced a more efficient version
that uses, on average $2k + \Theta(\cramped{(\log k)^2})$. Binary trees,
which are enumerated by the Catalan numbers~\cite{Stanley15}, are in
bijection with Dyck words, which are balanced words containing as many
$0$'s as $1$'s. So Bacher~\etal's random tree generation algorithm can be
used to produce a balanced word of size $2k$ using very few extra bits.

\paragraph{Rationale.}

The idea behind using a balanced word is that it is more efficient, in
average number of bits.

Indeed, splitting processes (repeatedly randomly partition $n$ elements
until each is in its own partition), are well known to require
$n\log_2 n + O(n)$ bits on average---this is the path length of a random
trie~\cite{FlSe09}. The linear term comes from the fact that when
processes are partitioned in two subsets, these subsets are not of equal
size (which would be the optimal case), but can be very unbalanced;
furthermore, with small probability, it is possible that all elements
remain in the same set, especially in the lower levels.

On the other hand, if we are able to partition the elements into two
equal-sized subsets, we should be able to circumvent this issue. This idea
is useful here, and we believe, would be useful in other contexts as well.

\paragraph{Disadvantages.} The advantage is that using balanced words
allows to for a more efficient and sparing use of random bits (and since
random bits cost time to generate, this eventually translates to savings
in running time). However this requires a linear amount of auxiliary
space; for this reason, our BalancedShuffled algorithm is generally slower
than the other, in-place algorithms.

\subsection{Correctness}

Denote by $S_n$ the symmetric group containing all permutations of size $n$. Let $C(n,k)$ be the set of all words of length $n$ on the alphabet $\{0,1\}$ containing $k$ $0$'s and $n-k$ $1$'s. We have $|S_n| = n!$ and $C(n,k) = \binom{n}{k}$.

We first prove the following lemma:

\begin{lemma}
  Assume we have a list of $n$ elements and a list of $m$ other elements.
  Shuffle both of them uniformly at random independently (i.e. sample an
  element in $S_n$ and an element in $S_m$ independently and uniformly at
  random). Now sample a word in $C(n+m,n)$ uniformly at random. With the
  process from Bacher~\etal~\cite{BaBoJa14}, we obtain a list of size
  $n+m$ that is a uniformly sampled random permutation of the $n+m$
  elements.
\end{lemma}

\begin{proof}
  We have defined a function

\begin{equation}
F : S_n \times S_m \times C(n+m,n) \rightarrow S_{n+m}.
\end{equation}
F is a surjection, because any given permutation of the $n+m$ elements can
be obtained by choosing adequate permutations of size $n$ and $m$, as well
as an adequate word in $C(n+m,n)$. Moreover, we have

\begin{equation}
\left|S_n \times S_m \times C(n+m,n)\right| = n! m! \binom{n+m}{n} = (n+m)! = |S_{n+m}|
\end{equation}
where $|\cdot|$ denotes the cardinality of a set. This implies that $F$ is
actually bijective. Thus, any element of $S_{n+m}$ has the same
probability of occurring, as it is obtained by a unique element of
$S_n \times S_m \times C(n+m,n)$. [end of proof]

If we use the ``bottom-up'' approach (we start with lists containing only
one element and work our way up), it follows by induction that the final
list is indeed a uniformly sampled random permutation.

If we use the ``top-down'' approach, the starting list is a uniformly
sampled random permutation of the final list, thus the final list is a
uniformly sampled random permutation of the starting list (the inverse of
a uniformly sampled random permutation is still a uniformly sampled random
permutation).
\end{proof}

\subsection{Average number of random bits}

We now give an estimate of the average number of random bits used by our
algorithm to sample a random permutation of size $n$. Let $cost(2k)$
denote the average number of random bits used to sample a random balanced
word of length $2k$ (an element of $C(2k,k)$). For the sake of simplicity,
assume that we sample a random permutation of size $n = 2^m$. The average
number of random bits used is then
\begin{align*}\label{bitcomplexity}
\sum_{i=1}^m 2^{m-i} cost(2^i).
\end{align*}

For the algorithm we have
$cost(2k) = 2k + \Theta (\log^2 k)$~\cite{BaBoJa14}. Thus, the average
number of random bits used to sample a random permutation of size
$n = 2^m$ is
\begin{align*}
\sum_{i=1}^m 2^{m-i} \left(2^i + \Theta\left(\log^2 (2^{i-1})\right)\right) = m 2^m + \Theta \left(2^m \sum_{i=1}^m \frac{i^2}{2^i}\right) = m 2^m + \Theta (2^m)
\end{align*}
which can be rewritten as

\begin{align*}
n \log_2 n + \Theta (n).
\end{align*}

\section{Simulations}

The simulations were run on a computing cluster with 40 cores. The
algorithms were implemented in C, and their parallel versions were
implemented using the OpenMP library, and delegating the distribution of
the threading entirely to it. 

Our algorithm, \textsc{MergeShuffle}, 

\begin{figure}[H]
\begin{center}
\includegraphics[scale=0.8]{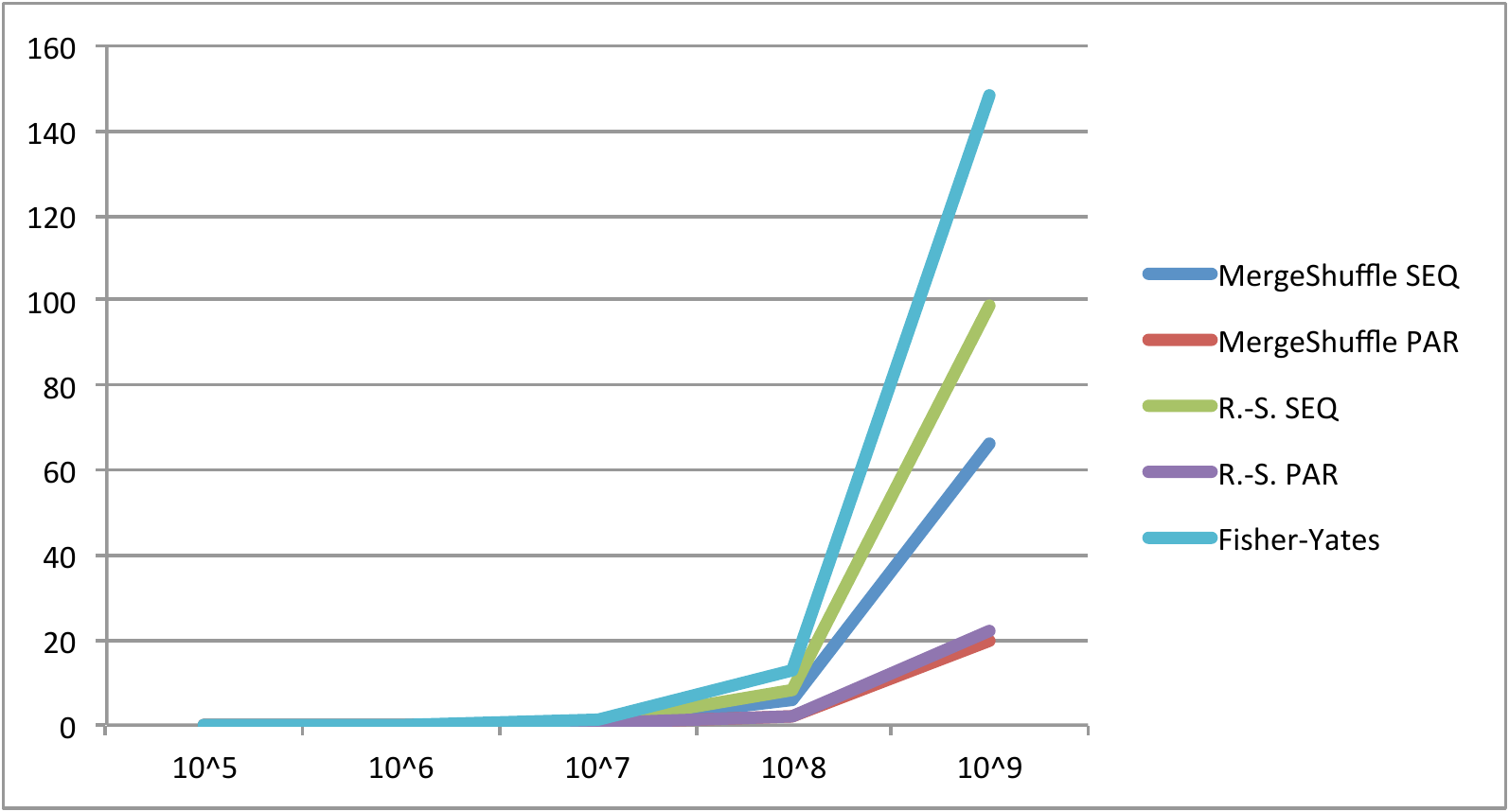}
\end{center}
\caption{Running times of several random permutation algorithms.
  Fisher-Yates shuffle, while extremely fast, gets slowed down once
  permutations are very large. Our parallel \textsc{MergeShuffle}
  algorithm is consistently faster than all algorithms, although the lead
  is not yet much compare to the Rao-Sandelius algorithm.}
\end{figure}

\begin{table}[t]
\begin{center}
\begin{tabular}{lrrrr}
  \toprule
  $n$	        &$10^5$	        &$10^6$	        &$10^7$	        &$10^8$\\\midrule
  Fisher-Yates	&1\,631\,434	&19\,550\,941	&229\,329\,728	&2\,628\,248\,831\\
  \textsc{MergeShuffle}	&1\,636\,560	&19\,686\,051	&231\,641\,075	&2\,650\,387\,993\\
  Rao-Sandelius	        &1\,631\,519	&19\,550\,449	&229\,327\,120	&2\,628\,251\,036\\
  \textsc{BalancedShuffle}	&1\,889\,034	&22\,046\,574	&	        &\\
  \bottomrule
\end{tabular}
\end{center}
\caption{Average number of random bits used by our implementation of various
  random permutation algorithms over 100 trials. (The current implementation of
  \textsc{BalancedShuffle} were in Python rather than C, and are prohibitively
  slow on larger permutations, but preliminary results show that it converges
  to an improved number of random bits.) }
\end{table}

\bibliographystyle{plain}
\bibliography{discrete-uniform,permutations,newg,additional}

\newpage
\appendix

\section{Code Listing for the \textsc{MergeSort} algorithm}

We reproduce here the most part of our algorithm, with some
OpenMP~\cite{DaEn98, Chandra01} hints. The full code can be obtained at
\url{https://github.com/axel-bacher/mergeshuffle}

\subsection{The merge procedure}
\begin{lstlisting}
// merge together two lists of size m and n-m
void merge(unsigned int *t, unsigned int m, unsigned int n) {
    unsigned int *u = t;
    unsigned int *v = t + m;
    unsigned int *w = t + n;

    // randomly take elements of the first and second list according to flips
    while(1) {
        if(random_bit()) {
            if(v == w) break;
            swap(u, v ++);
        } else
            if(u == v) break;
        u ++;
    }

    // now one list is exhausted, use Fisher-Yates to finish merging
    while(u < w) {
        unsigned int i = random_int(u - t + 1);
        swap(t + i, u ++);
    }
}
\end{lstlisting}

\subsection{The MergeSort algorithm itself}

\begin{lstlisting}
extern unsigned long cutoff;

void shuffle(unsigned int *t, unsigned int n) {
    // select q = 2^c such that n/q <= cutoff
    unsigned int c = 0;
    while((n >> c) > cutoff) c ++;
    unsigned int q = 1 << c;

    unsigned long nn = n;

    // divide the input in q chunks, use Fisher-Yates to shuffle them
    #pragma omp parallel for
    for(unsigned int i = 0; i < q; i ++) {
        unsigned long j = nn * i >> c;
        unsigned long k = nn * (i+1) >> c;
        fisher_yates(t + j, k - j);
    }

    for(unsigned int p = 1; p < q; p += p) {
        // merge together the chunks in pairs
        #pragma omp parallel for
        for(unsigned int i = 0; i < q; i += 2*p) {
            unsigned long j = nn * i >> c;
            unsigned long k = nn * (i + p) >> c;
            unsigned long l = nn * (i + 2*p) >> c;
            merge(t + j, k - j, l - j);
        }
    }
}
\end{lstlisting}

\end{document}